\documentclass[lettersize,journal]{IEEEtran} 

\usepackage{amssymb}
\usepackage{amsmath}
\usepackage{amsmath,bm}
\usepackage{amsthm}
\usepackage{graphicx}
\usepackage{cuted}
\usepackage{subfigure}
\usepackage{multirow}
\usepackage{makecell} 
\usepackage{cite}
\usepackage{enumerate}
\usepackage{enumitem}
\usepackage{subfigure}
\usepackage{stfloats}
\usepackage{algorithm}
\usepackage{algorithmic}  
\usepackage{color, soul}
\usepackage{hyperref}
\usepackage[labelsep=period]{caption}


\allowdisplaybreaks[4]

\DeclareMathOperator{\diag}{diag}

\begin{document}
	\setlength{\abovedisplayskip}{4pt}
	\setlength{\belowdisplayskip}{4pt}
	\newtheorem{theorem}{\bf~~Theorem}
	\newtheorem{remark}{\bf~~Remark}
	\newtheorem{property}{\bf~~Property}
	\newtheorem{observation}{\bf~~Observation}
	\newtheorem{definition}{\bf~~Definition}
	\newtheorem{lemma}{\bf~~Lemma}
	\newtheorem{preliminary}{\bf~~Preliminary}
	\newtheorem{proposition}{\bf~~Proposition}
	\newtheorem{comment}{\bf~~Comment}
	\renewcommand\arraystretch{0.9}
	\title{\LARGE{Holographic Integrated Sensing and Communication With Limited Radiation Amplitudes: How Many Quantization Bits Are Enough?}}
	\author{\normalsize \IEEEauthorblockN{
			{Shuhao Zeng}, \IEEEmembership{\normalsize  Member, IEEE}}

		\thanks{Shuhao Zeng is with Department of Electrical and Computer Engineering, Princeton University, NJ, USA (email: shuhao.zeng96@gmail.com).}	
	}
		\IEEEaftertitletext{\vspace{-4ex}}
	\maketitle
\begin{abstract}
	As a promising solution for extremely large-scale arrays, reconfigurable holographic surfaces~(RHS) can be integrated with integrated sensing and communication (ISAC) to form the holographic ISAC paradigm, where significant beamforming gains provided by RHS can improve both communication and sensing performance. However, most existing works on holographic ISAC assume that RHS elements can control the amplitudes of radiated ISAC signals in a continuous manner, which is hard to implement in practice. In this paper, we investigate how the limited radiation amplitudes of the RHS influence ISAC system performance. Specifically, we first derive closed-form lower bounds for both communication rate and sensing SINR to establish a tight upper bound on the minimum radiation amplitude quantization bits. Based on this, we further explore how the communication-sensing performance tradeoff impacts the quantization bits. Numerical results validate our theoretical analysis.
\end{abstract}

\vspace{-.3cm}
\section{Introduction}
Integrated sensing and communications~(ISAC) is widely regarded as a key enabling technology for future sixth-generation~(6G) networks~\cite{Yang_6G_2019}. In contrast to conventional radar and communication systems that operate independently over separate frequency bands, ISAC integrates sensing and communication functionalities over shared spectrum and hardware platforms, thereby significantly improving spectral, energy, and cost efficiency~\cite{Wang_Gen_2024}.

To further enhance ISAC performance, recent studies have investigated the use of reconfigurable holographic surfaces~(RHS), leading to an emerging paradigm of holographic ISAC~\cite{HISAC_Zhang_2022}. An RHS is a representative metasurface architecture composed of densely packed subwavelength elements embedded in a waveguide. By tuning the bias voltages of onboard diodes, each RHS element adjusts the amplitude of the radiated wave to enable holographic beamforming. Compared with conventional phased arrays that require complex and costly hardware, RHS supports low-cost and energy-efficient beamforming based on simple diodes, making it particularly attractive for large-scale deployments. The resulting large radiation aperture of the RHS enables high directive gain, thereby enhancing ISAC performance~\cite{Hu_how_many_2022}.

Existing works on holographic ISAC mainly focus on transceiver design and holographic beamforming~\cite{Zeng_MC_2025,HISAC_trx_architecture,covert_gao_2025}. For example, in~\cite{Zeng_MC_2025}, to prevent sidelobes induced by mutual coupling within the RHS from interfering with sensing process, a mutual-coupling-aware beamforming algorithm is proposed to jointly optimize the RHS configuration and digital beamformer to suppress sidelobes. In~\cite{HISAC_trx_architecture}, the authors develop a transceiver co-design framework for holographic ISAC systems, where digital beamformer, transmit/receive RHS configurations, and receive filters are jointly optimized to improve ISAC performance. However, most works assume that the RHS elements can control the radiation amplitude of the transmit ISAC signals continuously, which is hard to implement in practice~\cite{Hu_how_many_2022,Peng_Satellite_2025}. Although some initial works have considered limited radiation amplitudes of the RHS for holographic ISAC systems, such as~\cite{HISAC_Zhang_2022}, these works only consider beamforming design given the radiation amplitude quantization bits, while how the quantization bits influence holographic ISAC performance remains unexplored.


In this paper, we investigate a holographic ISAC system, where the base station~(BS) utilizes an RHS to serve multiple users and a sensing target via discrete amplitude-controlled holographic beamforming. We aim to determine the minimum number of quantization bits required for the system performance, defined as a weighted sum of the communication rate and sensing signal-to-interference-plus-noise ratio~(SINR), to satisfy a prescribed threshold. This is challenging since unlike conventional RHS-aided communication systems, the minimum bit requirement is jointly determined by both communication and sensing performance under limited radiation amplitudes. To address this, we derive closed-form lower bounds for both communication rate and sensing SINR to establish a tight upper bound on the required bits. On this basis, we further explore how the communication-sensing performance tradeoff impacts the minimum bit requirement. Numerical results validate our theoretical analysis.
	\vspace{-0.2cm}
	\section{System Model}
	\label{sec_system_model}
	\subsection{Scenario Description}	
	As shown in Fig.~\ref{sysmodel}, we consider a holographic ISAC system that consists of $L$ mobile users denoted by $\mathbb{L}= \{1,\cdots,l,\cdots,L\}$, one target, and a BS equipped with an RHS and a multiple-input-multiple-output~(MIMO) antenna array. The RHS transmits the ISAC signals with multiple beams directed toward the users and targets. The communication users receive and decode the transmitted signals to extract communication information~\cite{Zeng_MC_2025}, while the MIMO antenna array captures the echo signals reflected by the targets for radar sensing~\cite{Qin_CamEdit_2025}. Further, we assume there are $W$ clutters in the environment, where leakage power from the sidelobes of the RHS can propagate to various clutters and the resulting echo signals can cause interference to sensing procedures. 
	
	\begin{figure}[!t]
		\centering
		\includegraphics[width=0.35\textwidth]{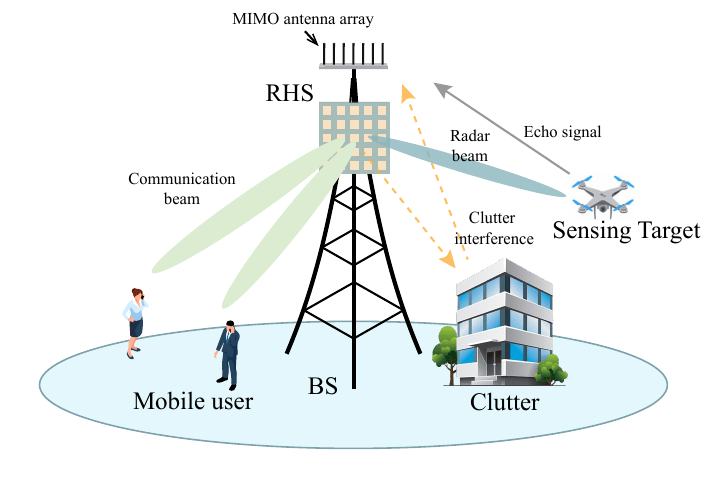}
		\caption{System model of an RHS-aided ISAC network.}
		\vspace{-4mm}
		\label{sysmodel}
	\end{figure}
\vspace{-.2cm}
\subsection{Holographic Beamforming Scheme}
The holographic beamforming scheme consists of two stages: digital beamforming at the BS and analog beamforming at the RHS. Specifically, $L$ communication data streams and $K$ radar waveforms are first precoded by the digital beamformers $\bm{V}_c \in \mathbb{C}^{T \times L}$ and $\bm{V}_s \in \mathbb{C}^{T \times K}$, respectively. The precoded baseband signals are then upconverted through $S$ RF chains, each of which is connected to one feed of the RHS. The signals injected into each feed propagate along the RHS as reference waves and sequentially excites one row of $M_{sub}$ radiation elements. By adjusting the biased voltages applied to onboard diodes, the radiation amplitudes $\{a_m\}$ of the RHS elements can be controlled to generate directional beams, thereby realizing RHS-based analog beamforming. Therefore, the signals transmitted by the RHS can be expressed as
\begin{equation}
	\bm{x}=\bm{A}\bm{F}(\bm{V}_c \bm{c}+\bm{V}_s \bm{s}),
\end{equation}
where $\bm{c}\in\mathbb{C}^{L\times1}$ and $\bm{s}\in\mathbb{C}^{K\times1}$ denote the communication symbols and radar waveforms, respectively, $\bm{A}=\diag\{a_1,\dots,a_M\}\in\mathbb{C}^{M\times M}$ collects the radiation amplitudes of all RHS elements, also referred to as holographic pattern, and $\bm{F}\in\mathbb{C}^{M\times S}$ captures the feed-to-element propagation. Here, we have $\bm{F}=\diag\{\bm{f}_1,\dots,\bm{f}_S\}$ with $\bm{f}_s=[f_{s,1},\dots,f_{s,M_{sub}}]^{T}$, where $f_{s,m}$ characterizes the phase shift and amplitude change introduced by the propagations from the $s$-th feed to the $m$-th RHS element in the $s$-th row.

In practice, the radiation amplitudes of RHS elements are limited and discrete rather than continuously tunable~\cite{Hu_how_many_2022}. By assuming that each element adopts $Q$-bit uniform quantization, the quantized radiation amplitude $\hat{a}_m$ of the $m$-th RHS element with limited amplitude values satisfies
\begin{align}
	\label{quant_level}
	\hat{a}_m \in \mathbb{A}_Q
	= \left\{ \frac{1}{2^{Q+1}}, \frac{3}{2^{Q+1}}, \ldots, \frac{2^{Q+1}-1}{2^{Q+1}} \right\}, \quad \forall n.
\end{align}
\vspace{-.5cm}
\subsection{Data Rate Model for Communication}
Let $\bm{h}_l = (h_{1,l}, \dots, h_{M,l})^T$ denote the channel from the RHS to the $l$-th mobile user. Define $\bm{V}$ as the combination of the digital beamformer $\bm{V}_c$ for communication signals and $\bm{V}_s$ for sensing signals, i.e., $\bm{V}=[\bm{V}_c,\bm{V}_s]\in \mathbb{C}^{T\times (L+K)}$. Then, the receive SINR of the $l$-th user can be expressed as~\cite{Zeng_MC_2025}		
\begin{align}
	\label{com_SINR}
	\gamma^{(C)}_l = \frac{|\bm{h}_l^{\mathsf{T}}\bm{A}\bm{F}\bm{v}_{l}|^2}{\sum\limits_{l'\le L+K,l'\neq l}|\bm{h}_l^{\mathsf{T}}\bm{A}\bm{F}\bm{v}_{l'}|^2+\sigma^2 }.
\end{align}
Here, $\bm{v}_l$ is the $l$-th column of the digital beamformer $\bm{V}$, $\sigma^2$ is receive noise power, and the term $|\bm{h}_l^{\mathsf{T}}\bm{A}\bm{F}\bm{v}_{l'}|^2$ represents inter-user interference and interference caused by the sensing signals when $l'\le L$ and $l'> L$, respectively.

Based on (\ref{com_SINR}), we derive the data rate for the $l$-th user, i.e.,
\begin{align}
	\label{data_rate}
	R_l=\log_2(1+\gamma^{(C)}_l).
\end{align}
The sum rate $R^{sum}$ over $L$ communication users, i.e., $R^{sum}=\sum_lR_l$, is utilized to characterize the communication performance of the holographic ISAC systems.
\vspace{-.3cm}
\subsection{SINR Model for Radar Sensing}
The ISAC waveforms radiated by the RHS are primarily directed toward the sensing target, where the target returns are received by the MIMO antenna array at the BS for sensing. However, due to sidelobes, a portion of the transmitted energy may also leak toward $W$ environmental clutters. The received echo signals from the clutters act as interference to the sensing process. To suppress noise and the interference caused by the clutters, the BS will apply a filter $\bm{o}$ to the echo signals. Therefore, the resulting received SINR is given by~\cite{Rang_waveform_2022,Quan_Siamese_2024}
\begin{align}
	\label{sensing_model}
	\gamma^{(S)}=
	\frac{\bm{o}^{H}\bm{H}_T\bm{A}\bm{G}\bm{A}^{H}\bm{H}_T^{H}\bm{o}}
	{\bm{o}^{H}\Big(\bm{H}_C\bm{A}\bm{G}\bm{A}^{H}\bm{H}_C^{H}
		+\sigma^{2}\bm{I}_N\Big)\bm{o}},	
\end{align}
where $\bm{H}_T\in \mathbb{C}^{N\times M}$ represents the target return channel, i.e., the channel from the $M$ RHS elements to the sensing target and back to the $N$ MIMO antennas at the BS. The matrix $\bm{H}_C\triangleq \sum\nolimits_w \bm{H}_w \in \mathbb{C}^{N\times M}$ denotes the aggregate clutter return channel, where $\bm{H}_w$ corresponds to the $w$-th clutter component. Moreover, $\bm{G}=\bm{F}\bm{V}\bm{V}^H\bm{F}^H$ characterizes the equivalent transmit covariance after digital precoding and feed-to-RHS propagation, and $\bm{I}_N$ is the identity matrix of dimension $N$. Here, we utilize the received sensing SINR in (\ref{sensing_model}) to characterize the radar sensing performance~\cite{Rang_waveform_2022}.



Based on the communication sum rate in (\ref{data_rate}) and sensing SINR in (\ref{sensing_model}), the overall ISAC performance is defined as~\cite{Wang_NOMA_2022}
\begin{align}
	\label{overall_performance}
	\mathcal{U} &= \frac{\rho}{R^{sum}_{\max}} \sum\nolimits_lR_l + \frac{1-\rho}{\gamma^{(S)}_{\max}} \gamma^{(S)}\\
	&\triangleq \rho^{(C)} \sum\nolimits_lR_l+\rho^{(S)}\gamma^{(S)},\notag
\end{align}
where $\rho \in [0,1]$ controls the tradeoff between communication and sensing, $R^{sum}_{\max}$ is the maximum sum rate, i.e., the sum rate of a communication-centric holographic ISAC system, and $\gamma^{(S)}_{\max}$ is the maximum sensing SINR, i.e., the sensing SINR of a sensing-centric holographic ISAC system. Here, $R^{sum}_{\max}$ and $\gamma^{(S)}_{\max}$ are utilized to normalize communication and sensing performance, respectively, such that the impact of quantization on both functionalities can be fairly compared.



\vspace{-.2cm}	
\section{Derivation of Required Radiation Amplitude Quantization Bits}
Since the number of radiation amplitude values of the RHS is finite in practice, the performance degradation is unavoidable. In this section, we analyze the effects of quantization bits on the performance of the considered holographic ISAC system, where a closed-form expression for the minimum number of quantization bits is derived.
\vspace{-.2cm}
\subsection{Quantization Bits Minimization Problem Formulation}
Denote by $\{a_m^*\}$ the optimal RHS configuration that maximizes the overall performance $\mathcal{U}$ in (\ref{overall_performance}) under ideal continuously tunable radiation amplitudes. Correspondingly, the optimal holographic beamforming matrix is $\bm{A}^*=\diag\{a_1^*,\dots,a_M^*\}\in\mathbb{C}^{M\times M}$. By substituting $\bm{A}^*$ into (\ref{overall_performance}), the resulting optimal performance of the holographic ISAC system is denoted by $\mathcal{U}^*$.

In practice, however, only a finite set of radiation amplitude levels is available, as indicated in (\ref{quant_level}). Therefore, each optimal amplitude $a_m^*$ is approximated by its nearest quantized value, denoted by $\hat{a}_m$, and the corresponding RHS matrix becomes $\hat{\bm{A}}=\diag\{\hat{a}_1,\dots,\hat{a}_M\}$. $Q$-bit quantization introduces a deviation from the optimal continuous solution, given by
\begin{align}
	\label{quantization_error_range}
	|a_m^*-\hat{a}_m| \le \frac{1}{2^{Q+1}}, \quad \forall n.
\end{align}

By substituting the quantized radiation amplitudes $\hat{\bm{A}}$ into (\ref{overall_performance}), the resulting ISAC performance is denoted by $\hat{\mathcal{U}}$, satisfying
\begin{align}
	\hat{\mathcal{U}} = \rho^{(C)} \sum\nolimits_l \hat{R}_l + \rho^{(S)} \hat{\gamma}^{(S)}.
	\label{performace_quantized}
\end{align}
Here, $\hat{R}_l$ and $\hat{\gamma}^{(S)}$ denote the communication rate of user $l$ and the sensing SINR, respectively, both evaluated with the quantized RHS configuration $\hat{\bm{A}}$.

According to (\ref{quantization_error_range}), quantization causes the radiation amplitudes $\hat{\bm{A}}$ of the RHS to deviate from their optimal values $\bm{A}^*$, thereby degrading system performance. To quantify ISAC performance degradation, we define the error $\epsilon$ as the ratio of the system performance with quantized amplitudes and that with continuous ones. 
\begin{align}
	\label{target_ratio}
	\epsilon=\frac{\hat{\mathcal{U}}}{\mathcal{U}^*}.
\end{align}

To ensure acceptable performance, we aim to find the minimum quantization bit $Q$ such that the ratio $\epsilon$ in (\ref{target_ratio}) exceeds a threshold $\epsilon_0$, which can be mathematically formulated as
\begin{equation}
	\label{problem_formulation}
	\min_{Q \in \mathbb{N}} Q \quad s.t. \epsilon \geq \epsilon_0.
\end{equation}
\vspace{-.5cm}
\subsection{Problem Solving Process}
It is non-trivial to determine the minimum number of required quantization bits, since it is difficult to obtain a closed-form relationship between the degraded system performance $\hat{\mathcal{U}}$ and quantization bits $Q$. To address this, we will derive tractable lower bounds for the quantized communication sum rate $\sum_l \hat{R}_l$ and the quantized sensing SINR $\hat{\gamma}^{(S)}$. 


Specifically, define 
$b_m^{l,l'} = h_{l,m}\bm{f}_m\bm{v}_{l'}$, where $h_{l,m}$ is the channel from the $m$-th RHS element to communication user $l$, and $\bm{f}_m$ records the reference wave propagation effect from different feeds to the $m$-th RHS element, and $\bm{v}_{l'}$ is the $l'$-th column of digital beamformer $\bm{V}$. Further define 
$c_{m,m'}^{l,l'} = \Re\!(b_m^{l,l'}(b_{m'}^{l,l'})^\dagger)$. Then, we have the following theorem.

\begin{theorem}
	\label{the_lower_bound_quant_rate}
	The sum rate $\sum_l\hat{R}_{l}$ under limited radiation amplitudes can be lower bounded by
	\begin{align}
		\label{lower_bound_quant_rate}
		&\sum\nolimits_l \hat{R}_{l}\ge \notag\\
		&\sum_l\! \log_2\!\Big(1 \!+ \!\frac{P_l^{(C)}}{I_l^{(C)} \!+\! \sigma^2 \!+\!
			\big(\sum\limits_{l'\le L+K, l'\neq l}
			\sum\limits_{m} \frac{c_{m,m}^{l,l'}}{12}\big)2^{-2Q}}\Big),
	\end{align}
	where $P_l^{(C)} \triangleq \big|\bm{h}_l^{\mathsf T}\bm{A}^*\bm{F}\bm{v}_l\big|^2$ and $I_l^{(C)} \triangleq \sum_{\substack{l'\le L+K \\ l'\neq l}} \big|\bm{h}_l^{\mathsf T}\bm{A}^*\bm{F}\bm{v}_{l'}\big|^2$ denote the desired communication signal power and the aggregated multi-user interference power at user $l$ under continuous RHS amplitudes, respectively. The additional term $\Big(\sum_{l'\neq l}\sum_m \frac{1}{12} c_{m,m}^{l,l'}\Big)2^{-2Q}$ captures the additional interference caused by amplitude quantization, while $\sigma^2$ denotes the noise power.
\end{theorem}

\begin{proof}
	See Appendix~\ref{app_lower_bound_quant_rate}.
\end{proof}

\vspace{-.4cm}
Similarly, we can derive a lower bound for the sensing SINR, as discussed below. Define $\bm{q}_m$ as the $m$-th row of $\bm{F}\bm{V}$, $\bm{d}_m^{(0)}=\bm{o}^{H}\bm{h}_m^{(0)}\bm{q}_m$, and $\bm{d}_m=\bm{o}^{H}\bm{h}_m\bm{q}_m$, where $\bm{h}_m^{(0)}$ and $\bm{h}_m$ are the $m$-th columns of the target channel matrix $\bm{H}_0$ and the clutter channel matrix $\sum_w\bm{H}_w$, respectively. Further define $e_{m,m'}=\Re(\bm{d}_m\bm{d}_{m'}^H)$. Then, we have


\begin{theorem}
	\label{the_lower_bound_sensing_SINR}
	The sensing SINR under limited radiation amplitudes can be lower bounded by
	\begin{align}
		\label{lower_bound_sensing_SINR}
		\hat{\gamma}_{S}\ge 
		\frac{P^{(S)}}{I^{(S)} +\sigma^2\bm{o}^H\bm{o}+ \left(\sum_{m}\frac{1}{12}e_{m,m}\right)2^{-2Q}},
	\end{align}
	where 
	$P^{(S)} \triangleq \bm{o}^{H}\bm{H}_0\bm{A}^*\bm{F}\bm{V}\bm{V}^H\bm{F}^H(\bm{A}^*)^H\bm{H}_0^H\bm{o}$ 
	represents the desired sensing signal power corresponding to the target echo under continuous RHS amplitudes, and 
	$I^{(S)} \triangleq \bm{o}^{H}\left(\sum_w\bm{H}_w\right)\bm{A}^*\bm{F}\bm{V}\bm{V}^H\bm{F}^H(\bm{A}^*)^H\left(\sum_w\bm{H}_w\right)^H\bm{o}$ 
	denotes the aggregated clutter-induced interference power under continuous RHS amplitudes. The additional term $\left(\sum_{m}\frac{1}{12}e_{m,m}\right)2^{-2Q}$ captures the interference increment caused by radiation amplitude quantization.
\end{theorem}
\begin{remark}
	The lower bounds for the quantized sum rate and quantized sensing SINR in (\ref{lower_bound_quant_rate}) and (\ref{lower_bound_sensing_SINR}) are achievable when the number $Q$ of quantization bits is sufficiently large.
\end{remark}

Based on the tractable lower bounds given in (\ref{lower_bound_quant_rate}) and (\ref{lower_bound_sensing_SINR}), we can derive an upper bound for the minimum required number of quantization bits, as shown in the following theorem.
\begin{theorem}
	\label{the_min_bit}
	The solution ${Q}^{min}$ to (\ref{problem_formulation}), i.e., the minimum required number of quantization bits, is upper bounded by
	\begin{align}
		\label{min_quant_bits}
		{Q}^{min}\le -\frac{1}{2}\log_2\zeta_0\triangleq \widetilde{Q}^{min}.
	\end{align}
	Here, $\zeta_0$ is the unique solution of the following equation with respect to $\zeta$\footnote{In general, $\zeta_0$ can be derived through numerical methods. However, for some special cases, such as a sensing-centric ISAC system with $\rho=0$, or a single-user communication-centric ISAC system with $\rho=1$, $\zeta_0$ can be given in closed form with respect to system parameters.}, i.e., 
	\begin{align}
	&\rho^{(C)}\sum_l \log_2\!\Bigg(1+\frac{P_l^{(C)}}{I_l^{(C)}+\sigma^2+
		\Big(\sum\limits_{{l'\le L+K, l'\neq l}}\sum\limits_m \frac{1}{12}c_{m,m}^{l,l'}\Big)\zeta}\Bigg)\notag\\
	&+\rho^{(S)}\frac{P^{(S)}}{I^{(S)}+\left(\sum_m \frac{1}{12}e_{m,m}\right)\zeta}
	=\epsilon_0\mathcal{U}^*.
	\end{align}
\end{theorem}
\begin{remark}
	\label{remark_achievable}
   The upper bound derived in Theorem~\ref{the_min_bit} is tight according to simulation results, and thus can be utilized to characterize required quantization bits and facilitate further system performance analysis, as shown below.
\end{remark}

\vspace{-.4cm}
\section{Performance Analysis}
\vspace{-.1cm}
\subsection{Comparison between Communication and Sensing Performance}
\vspace{-.1cm}

In this part, we analyze whether the communication or sensing functionality is more sensitive to limited radiation amplitudes. For fairness, we compare the average communication SINR in a communication-centric holographic ISAC system, i.e., ${\gamma}^{(C)}_{avg}(\rho=1)=\frac{1}{L}\sum_l{\gamma}_l(\rho=1)$, with the sensing SINR in a sensing-centric system, i.e., $\gamma^{(S)}(\rho=0)$.

According to (\ref{lower_bound_quant_rate}), the ratio between the average communication SINR $\hat{\gamma}^{(C)}_{avg}$ under limited radiation amplitudes and ${\gamma}^{(C)}_{avg}$ under continuous amplitudes can be approximated as
\begin{align}
	\label{ratio_SNR_comm}
	\left.\frac{\hat{\gamma}^{(C)}_{avg}}{{\gamma}^{(C)}_{avg}}\right|_{\rho=1}\!
	\approx\!
	\frac{I_{avg}^{(C)}+\sigma^2}{
		I_{avg}^{(C)}+\sigma^2+
		\Big(\!\sum\limits_{l=1}^{L}
		\sum\limits_{\substack{l'\le L+K\\ l'\ne l}}
		\sum\limits_{m}\!
		\frac{c_{m,m}^{l,l'}}{12L}\,\Big)2^{-2Q}
	},
\end{align}
where $I_{avg}^{(C)}=\frac{1}{L}\sum_l\left(\sum_{l'\le L+K, l' \neq l} |\bm{h}_l^T \bm{A}^* \bm{F} \bm{v}_{l'}|^2\right)$ denotes the average inter-user interference and sensing-induced interference among the $L$ users. Since only the communication performance is optimized, i.e., $\rho=1$, the interference can be effectively suppressed, leading to $I_{avg}^{(C)}\rightarrow 0$~\cite{HISAC_Zhang_2022,Zeng_IOS_2022}.

Furthermore, according to (\ref{lower_bound_sensing_SINR}), the ratio between the sensing SINR $\hat{\gamma}^{(S)}$ under limited radiation amplitudes and $\gamma^{(S)}$ under continuous radiation amplitudes can be expressed as
\begin{align}
	\label{ratio_SNR_sensing}
	\left.\frac{\hat{\gamma}^{(S)}}{\gamma^{(S)}}\right|_{\rho=0}=\frac{I^{(S)}+\sigma^2\bm{o}^H\bm{o}}{I^{(S)}+\sigma^2\bm{o}^H\bm{o}+(\sum_{m}\frac{1}{12}e_{m,m})(2^{-Q})^2},
\end{align}
where $I^{(S)}=\bm{o}^{H}(\sum_w\bm{H}_w)\bm{A}^*\bm{F}\bm{V}\bm{V}^H\bm{F}^H(\bm{A}^*)^H(\sum_w\bm{H}_w)^H\bm{o}$ represents the clutter-induced sensing interference. Since only sensing performance is considered here, i.e., $\rho=0$, the clutter interference is effectively mitigated, yielding $I^{(S)}\rightarrow 0$~\cite{Zeng_MC_2025}.

By comparing (\ref{ratio_SNR_comm}) and (\ref{ratio_SNR_sensing}), and noting that the interference terms approach zero, i.e., $I_{avg}^{(C)}\rightarrow 0$ and $I^{(S)}\rightarrow 0$, the relationship between $\frac{\hat{\gamma}^{(C)}_{avg}}{{\gamma}^{(C)}_{avg}}$ and $\frac{\hat{\gamma}^{(S)}}{\gamma^{(S)}}$ mainly depends on the relationship between $\frac{1}{L}\left(\sum_l\sum_{l'\le L+K, l' \neq l}  \sum_{m} \frac{1}{12} c_{m, m}^{l,l'}\right)$ and $\sum_{m}\frac{1}{12}e_{m,m}$. Substituting $c_{m, m'}^{l,l'}=\Re(b_{m}^{l,l'}(b_{m'}^{l,l'})^\dagger)$ with $b_{m}^{l,l'}=h_{l,m}\bm{f}_m\bm{v}_{l'}$, and $e_{m, m'}=\Re(\bm{d}_{m}\bm{d}_{m'}^H)$ with $\bm{d}_m=\bm{o}^{H}\bm{h}_m\bm{q}_m$, we observe that $\sum_{m}\frac{1}{12}e_{m,m} \ll \frac{1}{L}\left(\sum_l\sum_{l'\le L+K, l' \neq l}  \sum_{m} \frac{1}{12} c_{m, m}^{l,l'}\right)$, because the clutter return channel gain $|\bm{h}_m|$ is much weaker than the communication channel gain $|h_{l,m}|$, i.e., $|\bm{h}_m|\ll |h_{l,m}|$. Then, we obtain
\begin{align}
	\left.\frac{\hat{\gamma}^{(C)}_{avg}}{{\gamma}^{(C)}_{avg}}\right|_{\rho=1}< \left.\frac{\hat{\gamma}^{(S)}}{\gamma^{(S)}}\right|_{\rho=0}.
\end{align}
The above discussion leads to the following remark.

\begin{remark}
	\label{remark_sensitiveness}
	The sensing SINR of a sensing-centric holographic ISAC system is less sensitive to radiation amplitude quantization than the average communication SINR of a communication-centric holographic ISAC system.
\end{remark}

\vspace{-.4cm}
\subsection{Impact of Weight Factor $\rho$ on Minimum Required Bits}
In this part, we aim to analyze how the tradeoff between communication and sensing, characterized by weight factor $\rho$, influences the minimum required number of quantization bits $\widetilde{Q}^{min}$. We first consider two special cases, namely the communication-centric holographic ISAC system (i.e., $\rho=1$) and the sensing-centric holographic ISAC system (i.e., $\rho=0$). We denote the corresponding minimum required quantization bits by $\widetilde{Q}^{min}_1$ and $\widetilde{Q}^{min}_0$, respectively. According to Theorem~\ref{min_quant_bits}, the minimum required bits $\widetilde{Q}^{min}_1$ for the communication-centric holographic ISAC system satisfy
\begin{align}
	\label{min_bit_com}
	\frac{\sum_l\log_2(1+\hat{\gamma}_l(Q=\widetilde{Q}^{min}_1))}{\sum_l\log_2(1+{\gamma}_l)}=\epsilon_0.
\end{align}
In addition, the minimum required bits $\widetilde{Q}^{min}_0$ for the sensing-centric holographic ISAC system satisfy
\begin{align}
	\label{min_bit_sensing}
	\frac{\hat{\gamma}^{(S)}(Q=\widetilde{Q}^{min}_0)}{\gamma^{(S)}}=\epsilon_0.
\end{align}

Assume that the performance threshold is sufficiently high, i.e., $\epsilon_0$ is close to $1$. Then we have $\frac{\sum_l\log_2(1+\hat{\gamma}_l)}{\sum_l\log_2(1+{\gamma}_l)}\approx \frac{\hat{\gamma}^{(C)}_{avg}}{{\gamma}^{(C)}_{avg}}$. Based on Remark~\ref{remark_sensitiveness}, (\ref{min_bit_com}), and (\ref{min_bit_sensing}), we have the following.

\begin{remark}
	\label{remark_bits_cmp}
	When the performance threshold is sufficiently high, the sensing-centric holographic ISAC system requires fewer quantization bits than the communication-centric holographic ISAC system to achieve the performance threshold, i.e., $\widetilde{Q}^{min}_0<\widetilde{Q}^{min}_1$.
\end{remark}

\begin{remark}
	\label{remark_impact_lambda}
	When the performance threshold $\epsilon_0$ is sufficiently high, the minimum required number of bits in the holographic ISAC system increases as the weight factor grows, i.e., as the system places more emphasis on communication performance.
\end{remark}

\vspace{-.4cm}
\section{Simulation Results}
\label{sec_simulation}
\vspace{-.1cm}
In this section, we present simulation results to evaluate the performance of the proposed RHS-aided system. Following 3GPP standards and relevant literature~\cite{Rang_waveform_2022,Zhao_Sec_arxiv,Zeng_MC_2025,Zeng_Trajectory_2021}, the carrier frequency is set to $f = 30$~GHz, with a corresponding wavelength of $\lambda = 0.01$~m. The total transmit power at the BS is $P = 43$~dBm, and the AWGN power is $\sigma^2 = -96$~dBm. The BS is equipped with an RHS featuring $T = 6$ feed ports, each connected to an RF chain, with an internal reference wave attenuation factor $\alpha = 1$, an element spacing of $0.23\lambda$, and the number of RHS elements $M=72$. The reference wave model is the same as that in~\cite{Peng_Satellite_2025}, with propagation attenuation factor within the waveguide $\alpha=1$. For receiving echo signals, a MIMO antenna array with $N = 5$ elements and an inter-element spacing of $\lambda/2$ is utilized at the BS. For radar sensing purposes, the number of radar waveforms is set to $K = 1$. The number of clutters is set as $W=3$.

We assume that communication users, sensing targets, and clutter sources are randomly and uniformly distributed within specific spatial ranges. Specifically, the distance for communication users is uniformly distributed in $[100, 200]$~m, while targets and clutter are located within distances of $[30, 50]$~m and $[50, 70]$~m, respectively. For all entities, the elevation angles $\theta$ and azimuth angles $\phi$ are uniformly distributed in $[\pi/3, 2\pi/3]$ and $[-\pi/3, \pi/3]$, respectively. Both the communication and sensing echo channels are modeled as free-space line-of-sight~(LoS) propagation paths.


\begin{figure*}[!t]
	\centering
	\subfigure[]{
		\begin{minipage}[b]{0.27\textwidth}
			\centering
			\includegraphics[width=\textwidth]{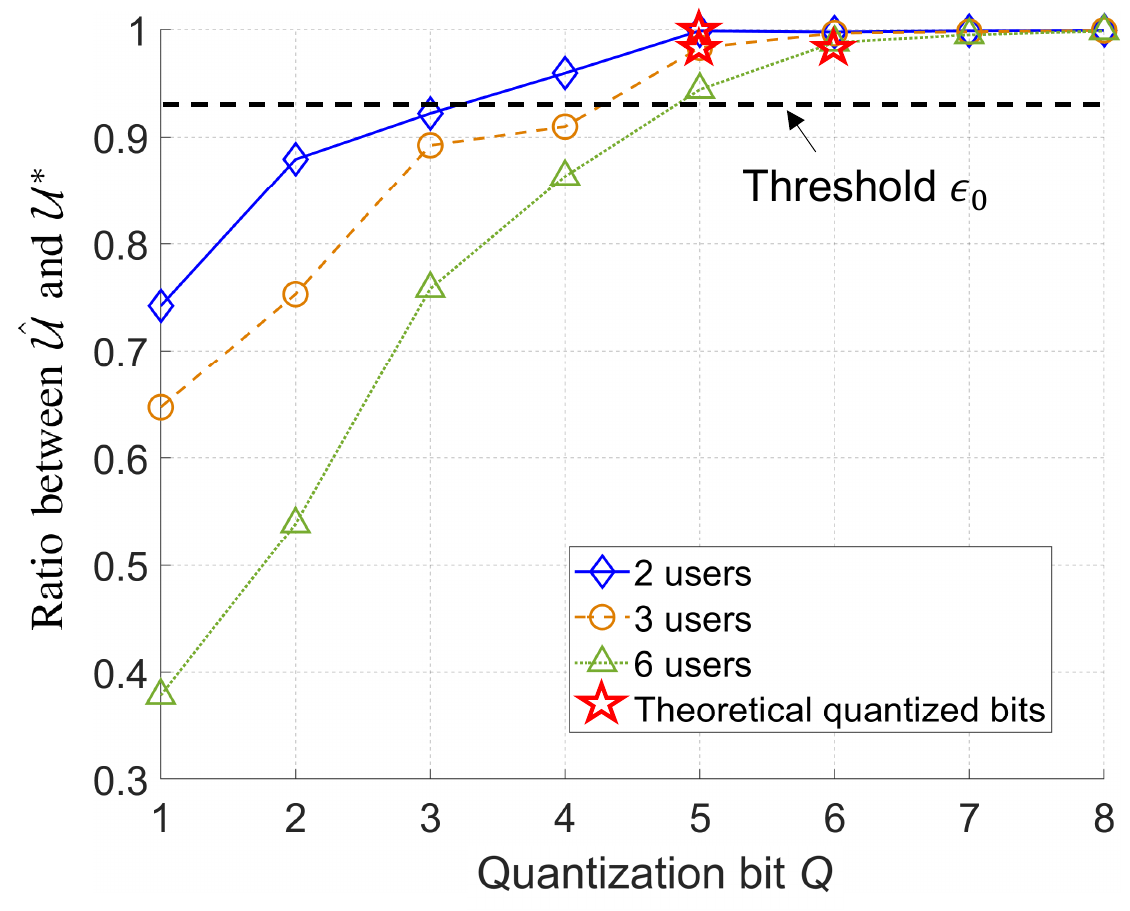}
			\vspace{-.6cm}
			\label{fig_result_L_all_v4}
	\end{minipage}}
	\subfigure[]{
		\begin{minipage}[b]{0.30\textwidth}
			\centering
			\includegraphics[width=\textwidth]{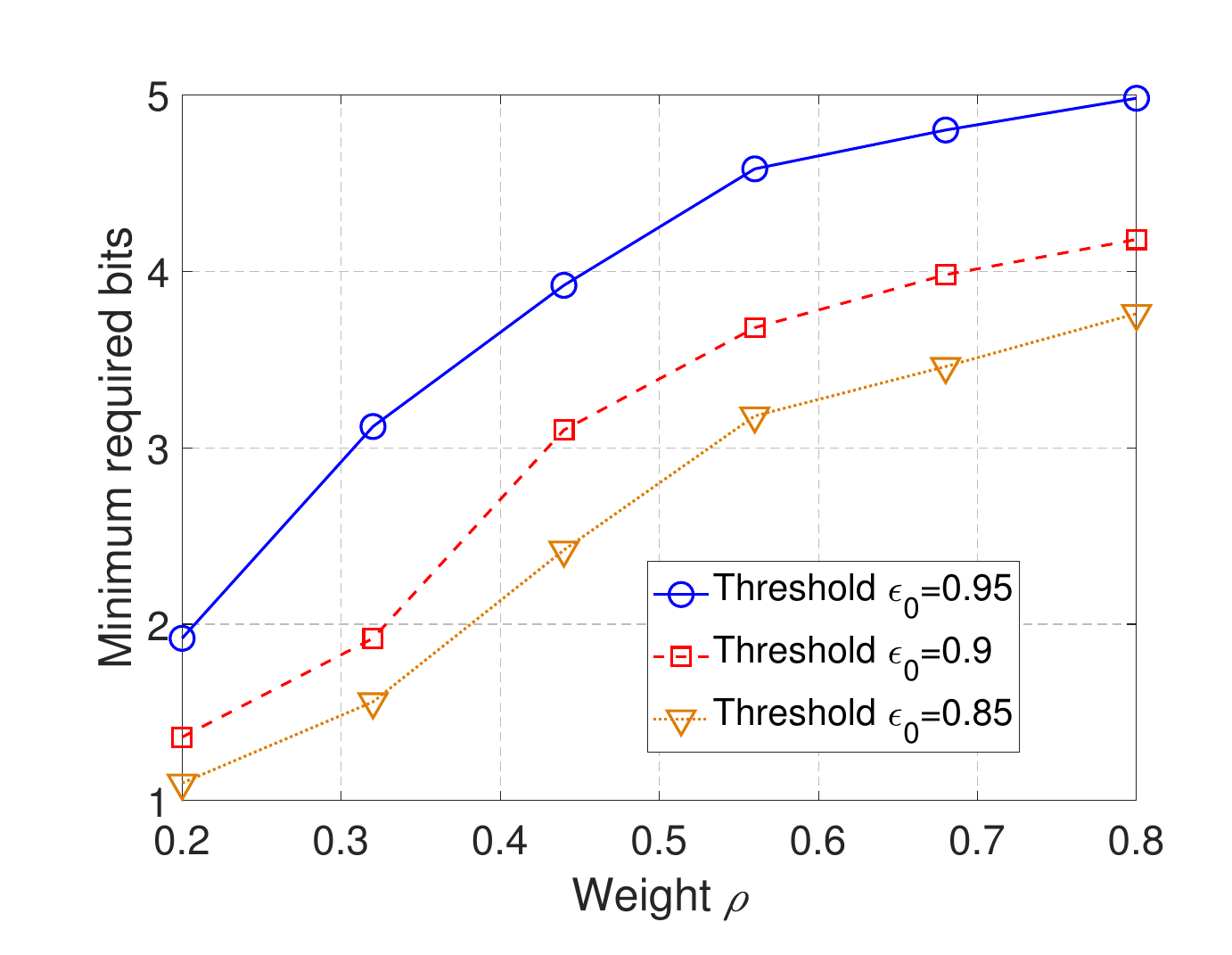}
			\vspace{-.6cm}
			\label{fig_result_quant_bit_vs_weight}
	\end{minipage}}
	\subfigure[]{
		\begin{minipage}[b]{0.29\textwidth}
			\centering
			\includegraphics[width=\textwidth]{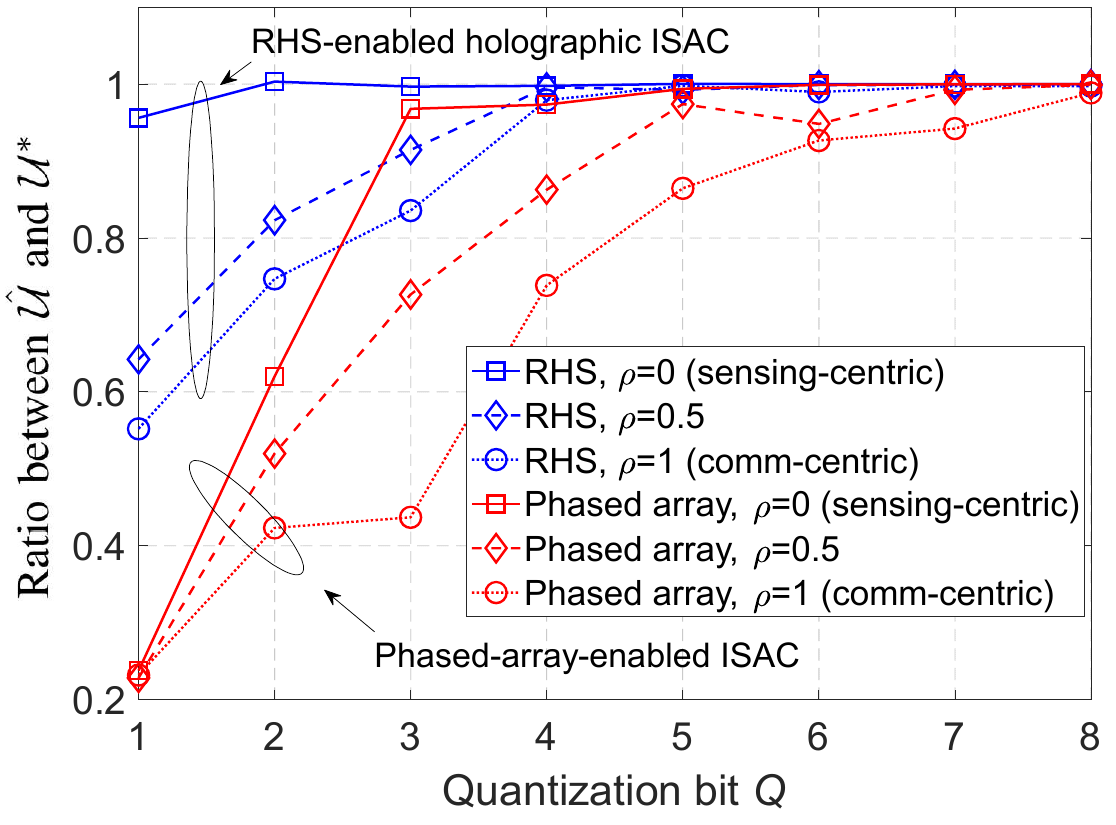}
			\vspace{-.6cm}
			\label{fig_result_cmp_RHS_vs_PA}
	\end{minipage}}
	\vspace{-.2cm}
	\caption{Numerical results: (a) Performance degradation $\epsilon=\hat{\mathcal{U}}/\mathcal{U}^*$ due to limited radiation amplitudes versus the number $Q$ of quantization bits, with $\rho=0.3$; (b) The minimum required quantization bits versus weight factor $\rho$, with $L=3$ communication users; (c) The comparison of quantization effects on the RHS-aided holographic ISAC systems and conventional ISAC systems based on phased arrays, with $L=2$ communication users and $W=10$ clutters. The hardware cost of the RHS is the same as that of the phased array~\cite{Zeng_MC_2025}.}
	\label{fig_overall_results}
	\vspace{-0.6cm}
\end{figure*}

Fig.~\ref{fig_result_L_all_v4} shows system performance ratio $\epsilon = \hat{\mathcal{U}}/\mathcal{U}^*$ versus quantization bits $Q$. We can find that the theoretical minimum quantization bits derived in (\ref{min_quant_bits}) are close to the realistic value of the minimum quantization bits, validating the accuracy of the analysis. Fig.~\ref{fig_result_L_all_v4} also indicates that systems with more communications users suffer more from quantization, thus requiring more bits to maintain performance.


Fig.~\ref{fig_result_quant_bit_vs_weight} illustrates the required minimum quantization bits versus the weight factor $\rho$, where the threshold is set high, i.e., $\epsilon_0\ge 0.85$. It can be observed that the required minimum quantization bits increase with weight factor $\rho$. This indicates that the holographic ISAC system requires more quantization bits when it pays more attention to the communication performance, which aligns with Remark~\ref{remark_impact_lambda}.



Fig.~\ref{fig_result_cmp_RHS_vs_PA} compares the quantization robustness of the RHS-enabled holographic ISAC system with that of a phased-array-based ISAC system. It is observed that the phased-array system is more sensitive to quantization and requires more bits to achieve the same performance threshold. This is because phased arrays rely on precise phase control for beamforming, and performance is inherently more sensitive to phase errors than to amplitude errors. Moreover, sensing-centric designs ($\rho=0$) require fewer quantization bits than communication-centric designs ($\rho=1$), consistent with Remark~\ref{remark_bits_cmp}.

\vspace{-0.3cm}
\section{Conclusion}
\vspace{-0.1cm}
This paper has investigated a holographic ISAC system enabled by an RHS with discrete amplitude control. Closed-form lower bounds for both communication rate and sensing SINR have been derived to establish a tight upper bound on the minimum required quantization bits. Based on this, we have further explored how the communication-sensing performance tradeoff impacts the quantization bit. Through theoretical analysis and numerical results, we can conclude that: 1) The sensing SINR for a sensing-centric holographic ISAC system is less sensitive to quantization effect than the average communication SINR for a communication-centric holographic ISAC system; 2) Under high performance threshold $\epsilon_0$, the minimum required bits tend to increase when the holographic ISAC system focuses more on communication performance; 3) Amplitude-control-based holographic ISAC systems are more robust to quantization effect than conventional ISAC systems based on phased arrays.

\vspace{-0.1cm}
\begin{appendices}
	\vspace{-.3cm}
	\section{Proof of Theorem~\ref{the_lower_bound_quant_rate}}
	\label{app_lower_bound_quant_rate}
	Let $\hat{a}_m=a_m^*+\delta_m$, where $\delta_m$ is the zero-mean quantization error. Substituting this into the SINR expression, both the desired signal and interference powers can be written as their continuous-amplitude counterparts plus perturbation terms involving $\delta_m$. For the desired signal, expanding $|\bm{h}_l^T\hat{\bm{A}}\bm{F}\bm{v}_l|^2$ yields a first-order term linear in $\delta_m$ and a second-order term quadratic in $\delta_m$. The first-order term vanishes in expectation, and removing the non-negative second-order term gives a lower bound on the desired signal. For the interference, a similar expansion shows that the first-order error terms vanish in expectation, while the second-order terms contribute an additional interference component $\big(\sum_{l'\neq l}\sum_m \frac{1}{12}c_{m,m}^{l,l'}\big)2^{-2Q}$.
\end{appendices}


\vspace{-.3cm}
\begin{thebibliography}{11}
	\vspace{-.2cm}
	\bibitem{Yang_6G_2019}
	P. Yang, Y. Xiao, M. Xiao, and S. Li, ``6G Wireless Communications: Vision and Potential Techniques," \emph{IEEE Netw.}, vol.~33, no.~4, pp.~70--75, Jul. 2019.
	
	\bibitem{Wang_Gen_2024}
	J. Wang, H. Du, D. Niyato, Z. Xiong, J. Kang, B. Ai, Z. Han, and D. I. Kim, ``Generative Artificial Intelligence Assisted Wireless Sensing: Human Flow Detection in Practical Communication Environments," \emph{IEEE J. Sel. Areas Commun.}, vol.~42, no.~10, pp.~2737--2753, Jun. 2024.
	
	\bibitem{HISAC_Zhang_2022}
	H. Zhang, H. Zhang, B. Di, M. D. Renzo, Z. Han, H. V. Poor, and L. Song, ``Holographic Integrated Sensing and Communication," \emph{IEEE J. Sel. Areas Commun.}, vol.~40, no.~7, pp.~2114--2130, Jul. 2022.
	
	\bibitem{Hu_how_many_2022}
	X. Hu, R. Deng, B. Di, H. Zhang, and L. Song, ``Holographic Beamforming for Ultra Massive MIMO With Limited Radiation Amplitudes: How Many Quantized Bits Do We Need?" \emph{IEEE Commun. Lett.}, vol.~26, no.~6, pp.~1403--1407, Jun. 2022.
	
	\bibitem{Zeng_MC_2025}
	S. Zeng, H. Zhang, B. Di, H. Zhang, Z. Shao, Z. Han, H. Vincent Poor, and L. Song, ``Holographic Beamforming for Integrated Sensing and Communication With Mutual Coupling Effects," \emph{IEEE J. Sel. Areas Commun.}, early access, 2025.
	
	\bibitem{HISAC_trx_architecture}
	J. Yu, H. Zhang, B. Di, and L. Song, ``Joint Transmit and Receive Beamforming for Holographic ISAC With Leakage Power Constraint," in \emph{Proc. IEEE Veh. Technol. Conf. (VTC-Fall)}, Chengdu, China, Oct. 2025.
	
	\bibitem{covert_gao_2025}
	W. Gao, Z. Xie, Y. Wang, Y. Yao, H. Jiang, and F. Shu, ``Covert Beamforming Design for Holographic Integrated Sensing and Communication With Imperfect CSI," \emph{IEEE Internet Things J.}, early access, 2025.
	
	\bibitem{Peng_Satellite_2025}
	Y. Peng, S. Zeng, S. Fu, and B. Di, ``On the Coverage Performance of Reconfigurable Holographic Surfaces-Aided LEO Satellite Communication Networks," \emph{IEEE Trans. Veh. Technol.}, early access.
	
	\bibitem{Qin_CamEdit_2025}
	X. Qin, Z. Wang, F. Li, H. Chen, R. Pei, W. Li, and X. Cao, ``CamEdit: Continuous Camera Parameter Control for Photorealistic Image Editing," in \emph{Proc. 39th Annu. Conf. Neural Inf. Process. Syst. (NeurIPS)}, San Diego, CA, USA, 2025.
	
	
	\bibitem{Rang_waveform_2022}
	R. Liu, M. Li, Q. Liu, and A. Lee Swindlehurst, ``Joint Waveform and Filter Designs for STAP-SLP-Based MIMO-DFRC Systems," \emph{IEEE J. Sel. Areas Commun.}, vol.~40, no.~6, pp.~1918--1931, Jun. 2022.
	
	\bibitem{Quan_Siamese_2024}
	Y. Quan, X. Qin, T. Pang, and H. Ji, ``Siamese Cooperative Learning for Unsupervised Image Reconstruction From Incomplete Measurements," \emph{IEEE Trans. Pattern Anal. Mach. Intell.}, vol.~46, no.~7, pp.~4866--4879, 2024.
	
	
	\bibitem{Wang_NOMA_2022}
	Z. Wang, Y. Liu, X. Mu, Z. Ding, and O. A. Dobre, ``NOMA Empowered Integrated Sensing and Communication," \emph{IEEE Commun. Lett.}, vol.~26, no.~3, pp.~677--681, Mar. 2022.
	
	\bibitem{Zeng_IOS_2022}
	S. Zeng, H. Zhang, B. Di, Y. Liu, M. Di Renzo, Z. Han, H. V. Poor, and L. Song, ``Intelligent Omni-Surfaces: Reflection-Refraction Circuit Model, Full-Dimensional Beamforming, and System Implementation," \emph{IEEE Trans. Commun.}, vol.~70, no.~11, pp.~7711--7727, Nov. 2022.
	
	\bibitem{Zhao_Sec_arxiv}
	C. Zhao, J. Wang, R. Zhang, D. Niyato, H. Du, Z. Xiong, D. I. Kim, and P. Zhang, ``SecDiff: Diffusion-Aided Secure Deep Joint Source-Channel Coding Against Adversarial Attacks," arXiv:2511.01466, 2025.
	
	\bibitem{Zeng_Trajectory_2021}
	A. Hassanien, M. G. Amin, E. Aboutanios, and B. Himed, ``DualFunction Radar Communication Systems: A Solution to the Spectrum
	Congestion Problem," \emph{IEEE Signal Process. Mag.}, vol.~36, no.~5, pp.~115-126, Sep. 2019.
	
	
%
%
%
%
%
%
%
%
%
%
%
%
%
%
%
%

\end{thebibliography}
\end{document}